\documentclass[10pt]{article}
\usepackage[letterpaper,margin=0.85in]{geometry}
\usepackage{times,natbib}
\setcitestyle{authoryear,round,citesep={;},aysep={,},yysep={;}}
\usepackage[T1]{fontenc}
\usepackage[utf8]{inputenc}
\usepackage{amsmath,amssymb,amsthm}
\usepackage{graphicx,booktabs,tabularx,multirow}
\usepackage{float}
\floatstyle{ruled}
\newfloat{algorithm}{tbp}{loa}
\floatname{algorithm}{Algorithm}
\usepackage{enumitem,microtype}
\usepackage{placeins}
\usepackage{hyperref,url}
\hypersetup{hidelinks,pdftitle={Evaluating Recoverability and Selective Prediction in Partially Observed PDE Systems},pdfauthor={}}
\graphicspath{{figures/}}
\newtheorem{proposition}{Proposition}
\newcommand{\scrs}{\mathrm{sCRS}}
\newcommand{\Rset}{\mathcal{R}_{\tau}}
\newcommand{\Bset}{\mathcal{B}_{\tau}}
\newcommand{\ind}{\mathbf{1}}
\DeclareMathOperator{\Cov}{Cov}
\DeclareMathOperator{\diag}{diag}
\DeclareMathOperator{\rank}{rank}
\DeclareMathOperator{\Unif}{Unif}
\DeclareMathOperator{\AUROC}{AUROC}
\setlist[itemize]{nosep,leftmargin=*}
\setlist[enumerate]{nosep,leftmargin=*}
\title{Recoverability Is a Subspace Property: A Benchmark for Certified State Estimation from Partial PDE Observations}
\newcommand{\namedversion}{1}
\author{Qingwei Dong$^{1,2}$ \quad Peng Zeng$^{1,2,*}$ \quad Guangxi Wan$^{1,2}$\\
Jiyuan Zhang$^{1,2,3}$ \quad Ruikai Liu$^{1,2,3}$ \quad Yuqi Liu$^{1,2}$\\[0.5em]
\small $^1$State Key Laboratory of Robotics, Shenyang Institute of Automation,\\
\small Chinese Academy of Sciences, Shenyang 110016, China\\
\small $^2$Key Laboratory of Networked Control Systems, Chinese Academy of Sciences,\\
\small Shenyang 110016, China\\
\small $^3$University of Chinese Academy of Sciences, Beijing 100049, China\\
\small $^*$Corresponding author: Peng Zeng, \texttt{zp@sia.cn}}
\date{}
\begin{document}
\maketitle
\begin{abstract}
When reconstructing the hidden state of a partial differential equation (PDE) system from partial observations, aggregate prediction error measures performance on a given data distribution but does not reveal how strongly the observations constrain each predicted direction. We introduce UniPDE-Bench, a direction-wise evaluation protocol that incorporates local observation geometry into state-estimation assessment, providing a reference for prediction recovery and confidence-based selection that is independent of the estimator. The protocol whitens the joint observation Jacobian by the noise covariance and normalizes it by a state metric. Its complete right singular basis represents joint variations of the state, which a relative sensitivity threshold partitions into retained and below-threshold directions. In this common basis, the protocol evaluates recovery and the agreement between prediction claims and the geometric partition, while recovery and abstention curves describe confidence-based selection at different claim coverages. In the simulated tasks and observation configurations studied here, confidence rules whose overall ranking exceeds chance can still exhibit below-random abstention on below-threshold directions at some high claim coverages. By separating empirical recovery from direction-selection quality, the protocol relates prediction performance to local observation sensitivity and provides an evaluation of partially observed state estimators beyond aggregate error.
\end{abstract}

\section{Introduction}
\label{sec:introduction}
When reconstructing the internal state or an unknown interface of a PDE system from partial observations, learned estimators use both the observations and statistical correlations in the training data. Aggregate reconstruction error measures prediction performance on a given distribution, but does not distinguish how strongly the observations constrain different state variations. A state direction here refers to a joint variation of the parameters: similar aggregate errors can correspond to different recovery outcomes along these directions. Comparing estimators therefore requires assessing both prediction accuracy and its relationship to the information supplied by the observations.

Consider $y=c_1+c_2$. Changing either coefficient alone changes the observation, whereas equal and opposite changes leave the sum unchanged. Coordinate-wise sensitivity therefore fails to reveal all ambiguities. If the data distribution exhibits strong correlations between the coefficients, an estimator may nevertheless predict both accurately. This example distinguishes predictability supported by a prior from distinguishability supplied by observations, motivating their comparison along joint directions.

Spectral analysis in classical inverse problems provides tools for characterizing observational constraints \citep{backus,tarantola}. Measurement-space and null-space decompositions have also been used to diagnose prior-induced structures in tomographic reconstruction \citep{hallucination}. Scientific machine learning benchmarks compare prediction performance across tasks and distributions \citep{pdebench,thewell}; selective prediction studies the trade-off between prediction coverage and error, while conformal prediction provides statistical coverage guarantees under appropriate assumptions \citep{geifman,angelopoulos}. These lines of work address observation maps, predictions, and statistical properties, respectively. We study their connection in direction-wise evaluation: how can local observation geometry provide a common reference for assessing both the directions a model recovers and those its confidence prioritizes?

We propose UniPDE-Bench, a direction-wise evaluation protocol whose reference is local observation geometry, independent of the estimator's predictions. The observation map matches the inputs actually used by the model, including observation histories and auxiliary information. The joint observation Jacobian describes the local response to simultaneous parameter changes. Noise whitening and state-metric normalization specify the comparison scales; a complete right singular basis retains all state directions, including the right null space. A relative sensitivity threshold then partitions them into retained and below-threshold directions. This partition describes relative local sensitivity near a working point. The below-threshold set includes nonzero low-sensitivity directions and does not identify globally unrecoverable directions.

After projecting predictions and ground truth into the same basis, the protocol checks whether errors on retained directions meet a recovery tolerance and whether the model makes prediction claims on below-threshold directions. Confidence determines which directions are claimed first; the remaining directions are abstained from, meaning that no prediction claim is made for them. Claim coverage is the fraction of directions claimed, distinct from the statistical coverage of prediction intervals. Sweeping claim coverage yields recovery-yield and below-threshold-abstention curves that separate empirical recovery from direction selection and reveal budget-dependent differences missed by a single operating point or overall ranking score. The relationship between the areas under these curves and their corresponding AUROCs clarifies the ranking information in the summary metrics.

We compare observation configurations and estimators using controlled local inverse problems and several PDE simulation tasks. The controlled experiments show that low prediction error under a correlated prior can coexist with a below-threshold subspace. Across models, accurate field reconstruction can also accompany direction selection that departs from the geometric reference. In the transport configuration studied here, conformal interval widths achieve an overall geometric ranking score above chance, yet their mean below-threshold abstention falls below the random reference at some high claim coverages. These empirical findings support reporting prediction recovery, observation geometry, and budget-dependent abstention together, rather than judging direction selection from aggregate error or an overall ranking score alone.

Our main contributions are as follows:
\begin{enumerate}
\item We combine an explicit observation model with complete right-singular-space analysis to establish a local geometric reference, enabling direction-wise comparisons between model predictions and the observation sensitivities associated with their actual inputs.
\item We define separate measures of prediction recovery, geometric agreement, and confidence-based selection, and derive the relationship between coverage-curve areas and AUROC, distinguishing recovery counts, ranking quality, and budget-dependent abstention.
\item Through controlled inverse problems and comparisons across configurations and models, we demonstrate differences between prediction accuracy and direction selection, and between overall ranking and performance at specific budgets, providing empirical evidence for interpreting partially observed state estimation.
\end{enumerate}

\section{Related Work}
\label{sec:related}
\paragraph{Inverse-problem geometry and subspace analysis.}
Resolution analysis and spectral regularization characterize how strongly data constrain different state components by examining directional responses of the observation map \citep{backus,tarantola}. Bayesian likelihood-informed subspaces and low-rank posterior approximations also account for prior geometry \citep{cui,spantini}; their interpretation cannot be directly equated with a sensitivity partition under an arbitrary state metric. Active subspace methods identify influential joint directions from gradient second moments averaged over an input distribution to construct reduced response surfaces \citep{constantine}. This distribution-averaged construction differs from the pointwise, noise-whitened observation Jacobian used here. UniPDE-Bench uses spectral geometry to construct local labels independent of model predictions, given an observation map, working point, state metric, and noise model. Its focus is to turn existing analytical tools into an evaluation reference against which observation sensitivity, prediction recovery, and direction selection can be compared.

\paragraph{Learned reconstruction and null-space diagnostics.}
Measurement-space and null-space decompositions have been used to analyze how reconstructions relate to observational constraints. \citet{hallucination} use this decomposition to diagnose prior-induced structures in tomographic reconstruction; \citet{gottschling} discuss kernel-induced limitations of inverse maps and the trade-off between accuracy and stability. These studies explain learned reconstruction through reconstructed components and properties of inverse maps, respectively. UniPDE-Bench applies this perspective to prediction claims: it relates empirical recovery to confidence rankings in a common directional basis and examines which directions are retained or abstained from at different claim budgets. Its local reference includes both the exact null space and nonzero low-sensitivity directions below a relative threshold. It evaluates predictions against the specified geometric partition without treating all prior-supported predictions as errors.

\paragraph{Scientific machine learning benchmarks.}
PDEBench and The Well support model comparisons across PDE tasks and data distributions \citep{pdebench,thewell}. PDEArena compares neural PDE surrogates and their generalization across equation parameters and time scales \citep{pdearena}. RealPDEBench pairs real-world measurements with numerical simulations to evaluate prediction and transfer between simulated and measured data \citep{realpdebench}. \citet{physbias} examine forecasting across physical regimes and distribution shifts, while the common task framework of \citet{seismicctf} includes sparse wavefield reconstruction. These studies address task coverage, generalization, and data realism. UniPDE-Bench instead evaluates recovery and confidence-based selection along directions defined by the local observation map. Its geometric labels must match the estimator's actual inputs, including observation histories and auxiliary channels. This analysis complements task-level evaluation; applying it to measured systems additionally requires a suitable differentiable observation model, state metric, and noise specification.

\paragraph{Evaluation protocols and reporting.}
Benchmark conclusions also depend on baseline strength and reporting choices. \citet{mcgreivy} document weak numerical baselines and reporting biases in machine learning for fluid-related PDEs. \citet{ctf} propose a common task framework with forecasting and state-reconstruction tasks, multiple metrics, and evaluation on hidden test sets. UniPDE-Bench addresses a complementary reporting issue within a specified observation configuration: aggregate error and a single ranking score can obscure differences between prediction recovery, local observation sensitivity, and abstention at particular claim budgets. Reporting these quantities separately adds an observation-geometric reference to existing prediction-accuracy metrics.

\paragraph{Selective prediction and uncertainty.}
Selective prediction studies the trade-off between prediction coverage and error by accepting or rejecting predictions \citep{elyaniv,geifman}. Conformal prediction provides statistical coverage guarantees under exchangeability and appropriate calibration conditions \citep{angelopoulos}, while ensembles provide uncertainty signals based on variation across model predictions \citep{lakshminarayanan}. Exchangeability is a statistical property of samples; injectivity of the observation map concerns whether observations distinguish states. These properties are not equivalent. We use joint variations of state parameters as the units of selection and separately evaluate whether prediction errors meet tolerance and whether confidence prioritizes retained directions. Interval widths and ensemble spread are used for empirical ranking analysis, focusing on their correspondence with local observation geometry rather than validating or refuting interval-coverage guarantees from ranking results.

\section{UniPDE-Bench: Constructing Observation Subspaces}
\label{sec:geometry}
UniPDE-Bench uses local observation geometry, independent of the estimator, as a reference for evaluating prediction recovery and confidence-based selection in a common directional basis. The protocol distinguishes the local sensitivity of observations to joint state variations from empirical prediction performance. It examines both the directions recovered by a model and the directions it prioritizes at different claim coverages, separating recovery capability from direction-selection quality.

As shown in Figure~\ref{fig:overview}, evaluation comprises observation-subspace construction, direction-wise assessment and confidence-based selection, and joint evaluation. First, the joint observation Jacobian corresponding to the model's actual inputs is whitened by the noise covariance and normalized by the state metric. Singular value decomposition constructs a complete right singular basis, and a relative sensitivity threshold separates retained from below-threshold directions. Next, model predictions and test ground truth are projected into this basis. Directional errors and a recovery tolerance determine successful-recovery labels on retained directions, while calibration or uncertainty records provide confidence rankings in the same basis. Finally, the ranking is held fixed as increasing numbers of directions are claimed in descending confidence order; all remaining directions are abstained from. Recovery-yield and below-threshold-abstention curves describe performance across claim coverages. Test ground truth is used only for scoring, not for confidence ranking. Geometric labels provide the evaluation reference rather than directly determining which directions are claimed.

\begin{figure}[t]
\centering
\includegraphics[width=\linewidth]{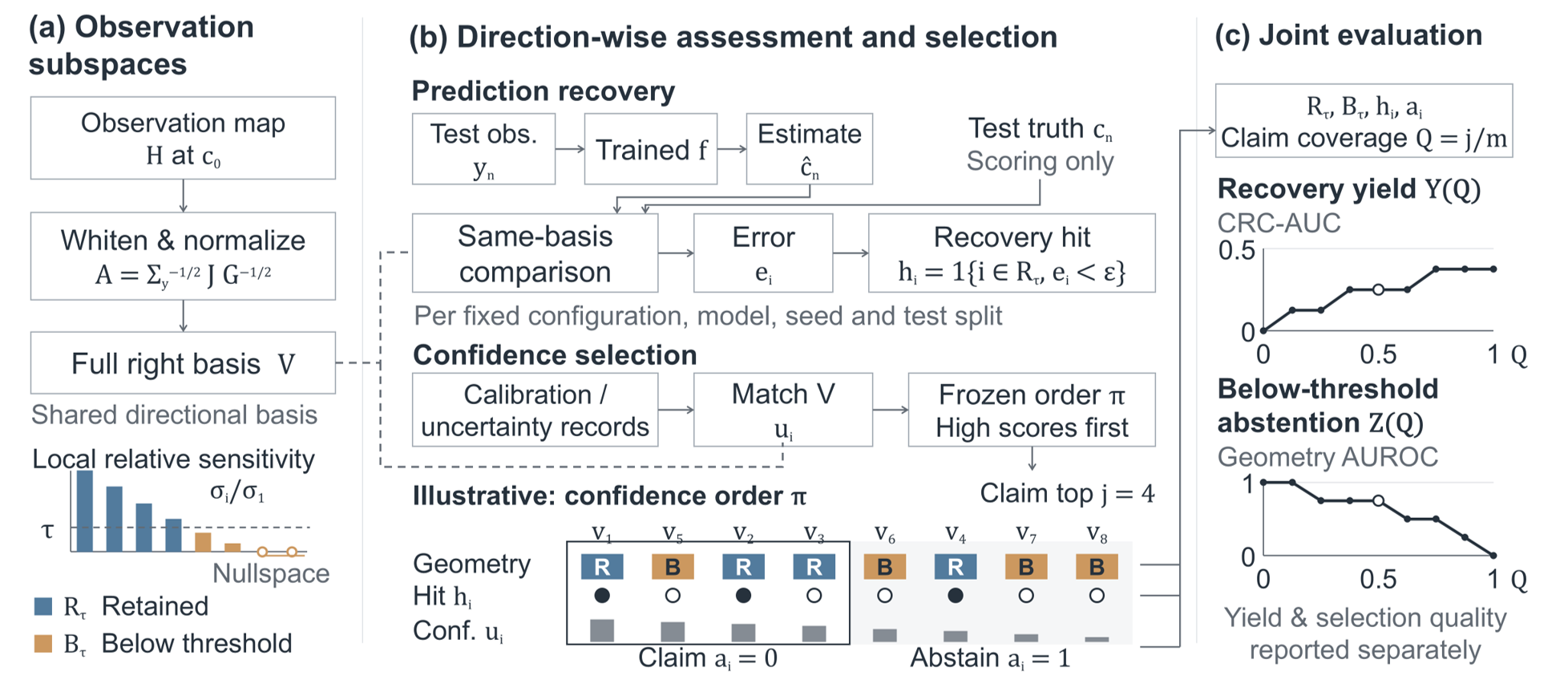}
\caption{UniPDE-Bench evaluation workflow. Observation geometry supplies a directional basis and sensitivity labels; model predictions and ground truth are used only to calculate directional errors and recovery labels. Confidence rankings and the claim budget jointly determine the claimed directions, which are then scored against the geometric labels. Below-threshold directions are not automatically excluded.}
\label{fig:overview}
\end{figure}

\subsection{Problem Setup and Observation Specification}
Let $c\in\mathcal{C}\subseteq\mathbb{R}^m$ denote the interface parameters to be estimated and $y\in\mathbb{R}^p$ the observations actually used by the model. Given a differentiable observation map $\mathcal{H}$ and a working point $c_0$, define
\begin{equation}
y=\mathcal{H}(c)+\eta,\qquad J=D\mathcal{H}(c_0),\qquad \Cov(\eta)=\Sigma_y\succ0.
\label{eq:observation}
\end{equation}
We use a positive-definite shape metric $G$, with $\|d\|_G^2=d^\top Gd$, and represent the local state by $z=G^{1/2}(c-c_0)$. The observation specification fixes the PDE, sensor functions, auxiliary channels, observation history, $c_0$, $G$, $\Sigma_y$, and the discretization and differentiation settings. A separate scoring specification fixes the relative sensitivity threshold, error tolerance, confidence rule, and aggregation procedure.

\subsection{Joint Sensitivity and Subspace Partition}
To compare joint perturbations on the scales defined by the state metric and observation noise, we define the whitened Jacobian and its complete right singular space:
\begin{equation}
A=\Sigma_y^{-1/2}JG^{-1/2},\qquad
A^\top A=V\diag(\sigma_1^2,\ldots,\sigma_m^2)V^\top.
\label{eq:whitened}
\end{equation}
Here $V=[v_1,\ldots,v_m]$ is orthogonal and $\sigma_1\geq\cdots\geq\sigma_m\geq0$. Numerical computations use the complete right singular basis. When $p<m$, all right-null-space directions must be retained, rather than using only the vectors returned by an economy-size SVD.

Given $0<\tau\leq1$ and $\sigma_1>0$, partition the directions into retained and below-threshold sets:
\begin{equation}
\begin{aligned}
\Rset&=\{i:\sigma_i/\sigma_1\geq\tau\},&
\Bset&=\{i:\sigma_i/\sigma_1<\tau\},\\
P_\tau&=\sum_{i\in\Rset}v_iv_i^\top,&
r_\tau&=|\Rset|,\qquad b_\tau=m-r_\tau.
\end{aligned}
\label{eq:partition}
\end{equation}
$P_\tau$ projects onto the retained subspace in local coordinates $z$. If $\sigma_1=0$, set $\Rset=\varnothing$, $\Bset=\{1,\ldots,m\}$, and $P_\tau=0$. Exact null-space directions satisfy $\sigma_i=0$, whereas $\Bset$ may also contain nonzero low-sensitivity directions. The relative threshold $\tau$ and numerical rank tolerance are recorded separately.

\subsection{Interpretation of Geometric Labels}
\begin{proposition}[Joint ambiguity]
Nonzero columns of a Jacobian do not imply that the joint map is injective. For example, the null space of $J=[1\;1]$ contains $(1,-1)^\top$. More generally, if $m>p$, then $\dim\ker J\geq m-p$. Coordinate-wise sensitivity therefore cannot replace joint-direction analysis.
\label{prop:joint}
\end{proposition}

\begin{proposition}[Indistinguishable states]
If $c_+,c_-\in\mathcal{C}$ satisfy $\mathcal{H}(c_+)=\mathcal{H}(c_-)$, the output of any deterministic estimator at this common observation has error at least $\tfrac12\|c_+-c_-\|_G$ for at least one of the two states.
\label{prop:indistinguishable}
\end{proposition}

Equation~\eqref{eq:partition} gives only relative local sensitivity labels; it does not establish the premise of Proposition~\ref{prop:indistinguishable} or global non-identifiability. Uniformly scaling the noise standard deviation leaves the partition unchanged, while absolute resolution also depends on perturbation amplitude and noise level. The choice of basis within a repeated-singular-value subspace must also be recorded, because directional errors and rankings may change under basis rotations. For example, $\mathcal{H}(c)=c^3$ has zero derivative at $c_0=0$ but remains injective. Finite perturbations $c_\pm=c_0\pm\delta G^{-1/2}v_i$ must belong to $\mathcal{C}$; a small nonzero observation difference supports only a local diagnosis. Proofs and finite-perturbation analysis are given in Appendix~\ref{app:theory}, and numerical sensitivity in Appendix~\ref{app:numerics}.

\section{Direction-Wise Selective Evaluation}
\label{sec:evaluation}
\subsection{Directional Predictions and Successful Recovery}
For a fixed configuration, model, random seed, and test split, project predictions $\widehat c_n$ and ground truth $c_n$ into the same right singular basis:
\[
q_n=V^\top G^{1/2}(c_n-c_0),\qquad
\widehat q_n=V^\top G^{1/2}(\widehat c_n-c_0).
\]
Define the directional signal scale, normalized error, and successful-recovery label as
\begin{equation}
\begin{aligned}
s_i^2&=\frac1N\sum_{n=1}^Nq_{ni}^2,\qquad
e_i=\frac{\sqrt{N^{-1}\sum_{n=1}^N(\widehat q_{ni}-q_{ni})^2}}{s_i},\\
h_i&=\ind\{i\in\Rset,\ e_i<\varepsilon\}.
\end{aligned}
\label{eq:recovery}
\end{equation}
When $s_i=0$, set $e_i=+\infty$ and $h_i=0$, and report the number of zero-signal directions. A label $h_i=1$ requires both retention and an error below tolerance; it is not determined by geometry alone. The choice $\varepsilon=1$ means outperforming the baseline that always predicts zero perturbation, not a worst-case error guarantee.

\subsection{Claims, Abstention, and Fixed-Budget Scores}
Let $a_i=1$ denote abstention on direction $i$ and $a_i=0$ a prediction claim. For $r_\tau,b_\tau>0$, report
\begin{equation}
\begin{aligned}
Q&=\frac1m\sum_{i=1}^m(1-a_i),&
R&=\frac1{r_\tau}\sum_{i=1}^m(1-a_i)h_i,\\
A_B&=\frac1{b_\tau}\sum_{i\in\Bset}a_i,&
H_B&=\sum_{i\in\Bset}(1-a_i).
\end{aligned}
\label{eq:fixed}
\end{equation}
$Q$ is claim coverage, $R$ the successfully recovered fraction of retained directions, $A_B$ the abstained fraction of below-threshold directions, and $H_B$ the number of below-threshold claims. These quantities describe output quantity, recovery capability, and agreement with the geometric partition. A ``claim'' means only that the prediction for that direction is accepted, not that identifiability is certified. The experiments use the soft composite score
\begin{equation}
\scrs=\max\!\left(0,1-\lambda\frac{H_B}{b_\tau}\right)
\frac{2RA_B}{R+A_B},\qquad\lambda=2.
\label{eq:scrs}
\end{equation}
The score is zero if $R+A_B=0$; if $b_\tau=0$, set $\scrs=R$; if $r_\tau=0$, set it to zero. Because Equation~\eqref{eq:scrs} applies a soft penalty to below-threshold claims, a positive score does not require $H_B=0$, so the components are also reported. Geometric gating directly sets $a_i=\ind\{i\in\Bset\}$, providing a reference for the recovery of a given predictor under the specified partition.

\subsection{Coverage Curves and Confidence Rankings}
Let $u_i$ denote directional confidence, with larger scores receiving higher claim priority. It must use the same directional basis as the predictions and must not be determined by test error $e_i$ or hit label $h_i$. Calibration error, interval width, and ensemble spread can provide ranking signals, each converted to confidence by a prespecified rule. We use a fixed direction ranking per configuration for calibration error and interval width.

Let $\pi$ sort $u_i$ in descending order. At each budget $j=0,\ldots,m$, claim the first $j$ directions and abstain from the remainder:
\begin{equation}
Q_j=\frac jm,\qquad
Y_j=\frac1m\sum_{\ell=1}^j h_{\pi(\ell)},\qquad
Z_j=\frac1{b_\tau}\sum_{\ell=j+1}^m\ind\{\pi(\ell)\in\Bset\}.
\label{eq:curves}
\end{equation}
$Y_j$ is recovery yield as a fraction of all directions, and $Z_j$ is the abstained fraction of below-threshold directions. When the denominators are nonzero, $Y_j=(r_\tau/m)R_j$ and $Z_j=A_{B,j}$. The two curves describe recovery yield and abstention separately. Under random ranking, their expectations are $Y(Q)=(k/m)Q$ and $Z(Q)=1-Q$, where $k=\sum_i h_i$. Integration uses piecewise-linear interpolation; tied groups contribute their expectation under a uniform random ordering.

For $0<k<m$, the normalized recovery-coverage area satisfies
\begin{equation}
\text{CRC-AUC}
=\frac{\int_0^1Y(Q)\,dQ-k/(2m)}{k(m-k)/(2m^2)}
=2\AUROC_h-1.
\label{eq:crc}
\end{equation}
Here $\AUROC_h$ treats successfully recovered directions as positives. This metric measures selection quality for a given set of prediction outcomes, rather than recovery quantity, and must therefore be reported with $k/m$ and the recovery curve. It is undefined when $k=0$ or $k=m$, although the curve and hit count are still reported.

To assess confidence against geometric labels alone, we use the classical pairwise-ranking interpretation of AUROC, with half credit for ties \citep{hanley}. Taking retained directions as positives and below-threshold directions as negatives, define the geometric AUROC when $r_\tau,b_\tau>0$:
\begin{equation}
\begin{aligned}
U_B&=\Pr(u_I>u_J)+\tfrac12\Pr(u_I=u_J),\\
I&\sim\Unif(\Rset),\qquad J\sim\Unif(\Bset),
\end{aligned}
\label{eq:geometric_auc}
\end{equation}
where $I$ and $J$ are drawn independently. The excess area of the below-threshold-abstention curve over the random reference satisfies
\begin{equation}
\begin{aligned}
D&=\int_0^1[Z(Q)-(1-Q)]\,dQ
=\frac{r_\tau}{m}\left(U_B-\frac12\right),\\
\frac{D}{D_{\mathrm{geom}}}&=2U_B-1,\qquad
D_{\mathrm{geom}}=\frac{r_\tau}{2m}.
\end{aligned}
\label{eq:excess}
\end{equation}
The pairwise interpretation in Equation~\eqref{eq:geometric_auc} is established knowledge. Our derivations in Appendix~\ref{app:auc} relate this classical AUROC interpretation to the specific coverage curves and normalizations defined here, yielding Equations~\eqref{eq:crc} and~\eqref{eq:excess}. An area and its corresponding AUROC express the same ranking information. An overall AUROC above chance does not ensure that $Z(Q)$ exceeds the random reference at every budget. Ranking by singular value gives $U_B=1$ as a geometric reference, not as evidence of learned uncertainty estimation.

\subsection{Pseudocode}
Algorithm~\ref{alg:protocol} first constructs geometric labels for each observation configuration, then evaluates each model, random seed, split, and selection rule separately. Test ground truth enters only the calculation of errors and hit labels, not confidence ranking. Interval widths are used only for empirical ranking analysis.

Each evaluation record stores identifiers for the observation and scoring specifications, the complete directional spectrum, basis $V$, directional errors and confidence scores, hit counts, and all curves, together with reasons for undefined metrics. Results for different models and seeds are computed separately before aggregation according to the specified rule. Reusing the same confidence vector does not provide additional geometric-ranking evidence.

\begin{algorithm}[t]
\caption{Direction-wise evaluation with UniPDE-Bench}
\label{alg:protocol}
\begin{minipage}{\linewidth}\small
\textbf{Input:} Observation specification $(\mathcal H,c_0,G,\Sigma_y)$; scoring specification $(\tau,\varepsilon)$ and numerical, tie-handling, and aggregation conventions; trained estimator $f$; test set $\{(y_n,c_n)\}_{n=1}^N$; confidence record $u$ with a directional-basis identifier.\par
\textbf{Output:} Directional records, fixed-budget metrics, coverage curves, and ranking scores.
\begin{enumerate}[label=\arabic*.,leftmargin=1.7em,itemsep=2pt,topsep=4pt]
\item Verify that $\mathcal H$ includes all history and auxiliary inputs used by $f$.
\item Compute $J=D\mathcal H(c_0)$ and $A=\Sigma_y^{-1/2}JG^{-1/2}$.
\item Compute the complete right singular basis $V$ of $A$, retaining all $m$ directions and their singular values.
\item Obtain $\Rset,\Bset,P_\tau$ from Equation~\eqref{eq:partition}; if $\sigma_1=0$, put all directions below threshold.
\item Compute $\widehat c_n=f(y_n)$ and project $c_n,\widehat c_n$ into $q_n,\widehat q_n$.
\item Compute $s_i,e_i,h_i$ using Equation~\eqref{eq:recovery}; if $s_i=0$, set $e_i=+\infty,h_i=0$.
\item Read $u_i$ matched to $V$ and freeze the confidence ranking $\pi$; do not update it using test errors.
\item \textbf{for} $j=0,\ldots,m$ \textbf{do}
\item \hspace*{1em}Set $a_{\pi(\ell)}=0$ for the first $j$ directions and $a_{\pi(\ell)}=1$ for the rest.
\item \hspace*{1em}Compute $Q_j,R_j,A_{B,j},H_{B,j},Y_j,Z_j$ using Equations~\eqref{eq:fixed} and~\eqref{eq:curves}.
\item \textbf{end for}; for tied groups, use expected curve contributions under uniform random ordering; integrate piecewise linearly.
\item Set $k=\sum_i h_i$; compute CRC-AUC if $0<k<m$, otherwise mark it undefined.
\item Compute $U_B,D,D/D_{\mathrm{geom}}$ if $r_\tau,b_\tau>0$; do not report $U_B$ if either class is absent.
\item Save specification identifiers, $V,\sigma,e,h,u,k,r_\tau,b_\tau$, curves, metrics, and reasons for undefined values.
\end{enumerate}
\textbf{Boundary conventions:} if $r_\tau=0$, set $R_j=0$; if $b_\tau=0$, $Z_j$ and $A_{B,j}$ are not applicable.
\end{minipage}
\end{algorithm}

\section{Experiments}
\label{sec:experiments}
We investigate three questions: which state directions joint observations constrain; whether low prediction error accompanies appropriate direction selection; and how confidence rankings behave at different claim coverages. We first analyze the local sensitivities of three reference operators, then use controlled inverse problems and model comparisons to distinguish prediction accuracy from observation geometry, and finally examine confidence rankings and the effect of the coverage budget.

\subsection{Experimental Setup}
\paragraph{Simulation tasks and configurations.}
We use a diffuse-domain battery-potential model (Battery), a low-P\'eclet-number advection--diffusion model (Melt-pool), and a high-P\'eclet-number transport model (Flooding). All three are interface-sensing simulation tasks. The interface is represented by a real Fourier expansion with $K=8$ and a constant term, giving $m=17$ coefficients. Reference diagnostics use a $64\times64$ grid and a relative sensitivity threshold $\tau=0.05$. Observation noise is normalized by channel scale at a strength of $1\%$, with an absolute floor.

Observation-geometry analysis uses the three reference configurations in Table~\ref{tab:geometry}; Battery includes two auxiliary channels. Neural-model evaluation uses five configurations: Battery M0, M2, and M8, Melt-pool M12, and Flooding M12. Observation specifications for the reference configurations and learning tasks are recorded separately. Direction-wise evaluation requires geometry to account for all inputs actually used by an estimator, including observation histories and auxiliary channels.

\paragraph{Estimators and repeated experiments.}
The controlled local inverse problems use ridge regression and truncated singular value decomposition (TSVD), with 1,600 training, 500 validation, and 800 test samples, repeated over five random seeds. Neural models comprise UniPDE, UniPDE-noalign, a GRU, a static MLP, FNO \citep{fno}, and DeepONet \citep{deeponet}. Each model family is trained on each configuration with five random seeds, yielding 150 training runs. Poseidon-T \citep{poseidon} adaptation and prediction-averaging experiments are reported separately in Appendix~\ref{app:poseidon}.

\paragraph{Metrics and comparisons.}
Prediction accuracy is measured by normalized root mean squared error (NRMSE). The controlled inverse problems and Poseidon-T evaluate interface-coefficient error, whereas the main comparison of the six neural model families evaluates field error. Selective performance is measured by $\scrs$, recovery-coverage area (CRC-AUC), and geometric AUROC, together with recovery-yield and below-threshold-abstention curves. Oracle gating directly uses geometric labels as a direction-selection reference, not as a learned confidence estimator. Results are computed separately for in-distribution (ID) and out-of-distribution (OOD) data. Table~\ref{tab:neural} first averages the five seeds within each model family and then takes a macro-average over the six families.

\subsection{Local Sensitivity of Joint Observations}
Coordinate-wise and joint-direction analyses give different observation diagnostics. In Table~\ref{tab:geometry}, coordinate-wise analysis detects no blind modes for Flooding, although its $13\times17$ Jacobian cannot have full column rank. Joint spectral analysis gives retained dimensions of 12, 9, and 12 for Battery, Melt-pool, and Flooding at $\tau=0.05$, with corresponding below-threshold dimensions of 5, 8, and 5 (Figure~\ref{fig:spectra}). Observable responses to independent changes in every coefficient therefore do not rule out ambiguities caused by joint coefficient variations.

\begin{table}[t]
\caption{Joint-Jacobian diagnostics for the three reference configurations ($m=17$, $\tau=0.05$, $64\times64$ grid). The below-threshold dimension is the state dimension minus the effective rank, not the exact null-space dimension. Coordinate-wise conclusions concern independent perturbations of individual coefficients.}
\label{tab:geometry}
\centering\small
\setlength{\tabcolsep}{3pt}
\begin{tabular*}{\linewidth}{@{\extracolsep{\fill}}lccccl@{}}
\toprule
Operator & $J$ size & \shortstack{Structural\\rank} & \shortstack{Effective\\rank} & \shortstack{Below\\threshold} & Coordinate-wise finding\\
\midrule
Battery & $14\times17$ & 13 & 12 & 5 & High-frequency gain decay\\
Melt-pool & $13\times17$ & 13 & 9 & 8 & High-frequency gain decay\\
Flooding & $13\times17$ & 13 & 12 & 5 & ``No blind modes''\\
\bottomrule
\end{tabular*}
\end{table}

For Flooding, 13 observations and 17 state coefficients imply at least four exact null-space directions in the local linear map. The thirteenth normalized singular value is 0.039, below the threshold 0.05, adding a low-sensitivity direction. The five below-threshold directions thus include both exact null-space and nonzero weak-sensitivity components. This partition describes local observation sensitivity rather than directly determining global identifiability of the nonlinear system.

\begin{figure}[t]
\centering
\includegraphics[width=\linewidth]{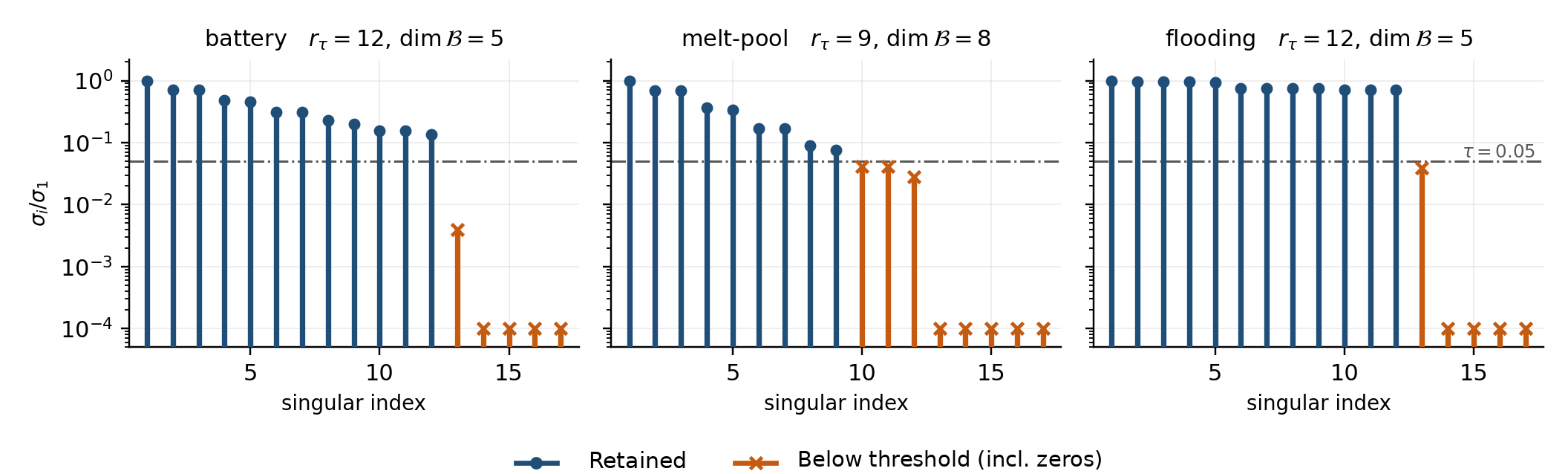}
\caption{Normalized joint singular-value spectra of the three reference operators. The dashed line marks the relative threshold $\tau=0.05$; retained dimensions are 12, 9, and 12, respectively. Padded zeros and values below the display range are placed at the lower axis limit; their display positions are not used to determine exact rank.}
\label{fig:spectra}
\end{figure}

Discretization affects the threshold partition differently across operators. On a $48\times48$ grid, their effective ranks are 12, 7, and 12; on a $64\times64$ grid, Melt-pool increases to 9 while the other two remain unchanged. This result calls for separate numerical-sensitivity checks near the threshold: a small matrix-norm change alone does not establish subspace convergence. Complete spectra and finite-difference diagnostics appear in Appendix~\ref{app:numerics}.

\subsection{Low-Error Prediction under Correlated Priors}
To distinguish prior-supported prediction from information supplied by observations, we construct controlled local inverse problems. Coefficients are sampled from a smooth correlated prior with a small full-rank residual component. Noisy observations are generated using the local Jacobian, and ridge regression is fitted. This construction retains joint ambiguities in the observation map while allowing the estimator to exploit statistical correlations between coefficients.

\begin{table}[t]
\caption{Prediction errors and selective scores in the controlled local inverse problems. Coefficient NRMSE is the mean $\pm$ standard deviation over five random seeds. Oracle gating uses geometric labels; TSVD + Oracle combines a spectrally truncated estimator with this gating rule.}
\label{tab:controlled}
\centering\small\setlength{\tabcolsep}{3.5pt}
\begin{tabular}{@{}lccccc@{}}
\toprule
Operator & \shortstack{Ridge OOD\\NRMSE} & \shortstack{No abstention\\$\scrs$} & \shortstack{Validation\\gate $\scrs$} & \shortstack{Oracle\\gate $\scrs$} & \shortstack{TSVD +\\Oracle $\scrs$}\\
\midrule
Battery & $0.669\pm0.007$ & 0.00 & 0.00 & 1.00 & 0.40\\
Melt-pool & $0.140\pm0.004$ & 0.00 & 0.00 & 1.00 & 0.875\\
Flooding & $0.054\pm0.001$ & 0.00 & 0.00 & 1.00 & 1.00\\
\bottomrule
\end{tabular}
\end{table}

Table~\ref{tab:controlled} reports ridge-regression OOD coefficient NRMSEs of $0.669\pm0.007$, $0.140\pm0.004$, and $0.054\pm0.001$ on Battery, Melt-pool, and Flooding. Flooding has the lowest prediction error, yet its local observation map still has five below-threshold directions. Meanwhile, both no abstention and validation-error gating yield zero $\scrs$, whereas geometric-reference gating scores 1.00 (Figure~\ref{fig:controlled}). Low distribution-averaged prediction error therefore does not automatically translate into direction selection aligned with observation geometry.

TSVD + Oracle scores 0.40, 0.875, and 1.00 on the three tasks. Thus, different estimators can retain different recovery performance even under geometric-reference gating. Geometric labels identify directions of relatively high local sensitivity, while prediction errors determine whether those directions meet the recovery tolerance. These aspects require separate evaluation.

\begin{figure}[t]
\centering
\includegraphics[width=\linewidth]{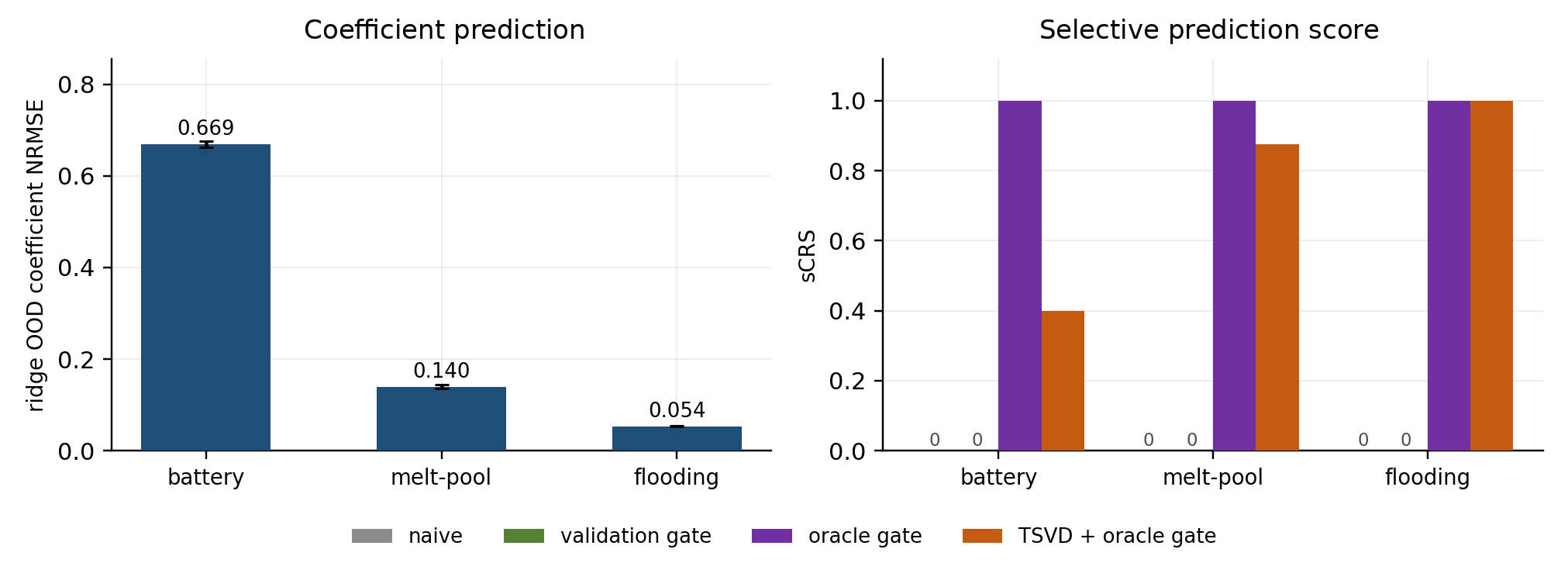}
\caption{Prediction accuracy and selective scores in the controlled inverse problems. Left: ridge-regression OOD coefficient NRMSE, with error bars showing standard deviations over five seeds. Right: $\scrs$ for no abstention, validation-error gating, Oracle gating, and TSVD + Oracle.}
\label{fig:controlled}
\end{figure}

\subsection{ID and OOD Performance of Neural Estimators}
The six neural model families further show that field-reconstruction accuracy and direction-wise selective scores are not interchangeable. Under ID conditions, macro-averaged field NRMSEs for the three Battery configurations are similar---0.048, 0.048, and 0.047---but validation-gated $\scrs$ values are 0.465, 0.345, and 0.049, while Oracle-gated scores are 1.000, 0.800, and 0.228 (Table~\ref{tab:neural}). Similar aggregate reconstruction errors do not imply similar directional recovery or selection.

Under distribution shift, macro-averaged field NRMSE rises from 0.057 to 0.306. Most of the change occurs in Melt-pool M12 and Flooding M12, where errors rise from 0.072 and 0.071 to 0.530 and 0.483, respectively, and both gated $\scrs$ values fall to zero. These evaluation records contain no hits meeting the directional recovery tolerance, indicating that changing the gating rule alone cannot compensate for inadequate recovery.

Figure~\ref{fig:neural} compares gating for 30 model--configuration combinations. Under ID conditions, Oracle gating scores higher than validation gating in all 30 combinations; under OOD conditions, it does so in 18. The remaining 12 combinations are Melt-pool M12 and Flooding M12, for which both scores are zero. For example, FNO on Flooding M12 under ID conditions achieves field NRMSE $0.039\pm0.002$, yet validation-gated $\scrs$ is $0.179\pm0.183$, compared with $0.928\pm0.026$ under Oracle gating. Even accurate field prediction can therefore accompany validation-error-based direction selection that departs from the geometric reference.

\begin{table}[t]
\caption{Field-reconstruction errors and gating scores for six neural model families. Each configuration first averages the five seeds within each model family, then macro-averages over the six families. Lower NRMSE is better; the two $\scrs$ columns correspond to validation-error and geometric-reference gating.}
\label{tab:neural}
\centering\small
\begin{tabular}{@{}llccc@{}}
\toprule
Split & Configuration & Field NRMSE & \shortstack{Validation-gated\\$\scrs$} & \shortstack{Oracle-gated\\$\scrs$}\\
\midrule
ID & Battery M0 & 0.048 & 0.465 & 1.000\\
ID & Battery M2 & 0.048 & 0.345 & 0.800\\
ID & Battery M8 & 0.047 & 0.049 & 0.228\\
ID & Melt-pool M12 & 0.072 & 0.559 & 0.668\\
ID & Flooding M12 & 0.071 & 0.281 & 0.767\\
ID & Macro-average & 0.057 & 0.340 & 0.692\\
\midrule
OOD & Battery M0 & 0.168 & 0.438 & 0.900\\
OOD & Battery M2 & 0.170 & 0.345 & 0.800\\
OOD & Battery M8 & 0.180 & 0.050 & 0.244\\
OOD & Melt-pool M12 & 0.530 & 0.000 & 0.000\\
OOD & Flooding M12 & 0.483 & 0.000 & 0.000\\
OOD & Macro-average & 0.306 & 0.166 & 0.389\\
\bottomrule
\end{tabular}
\end{table}

\begin{figure}[t]
\centering
\includegraphics[width=\linewidth]{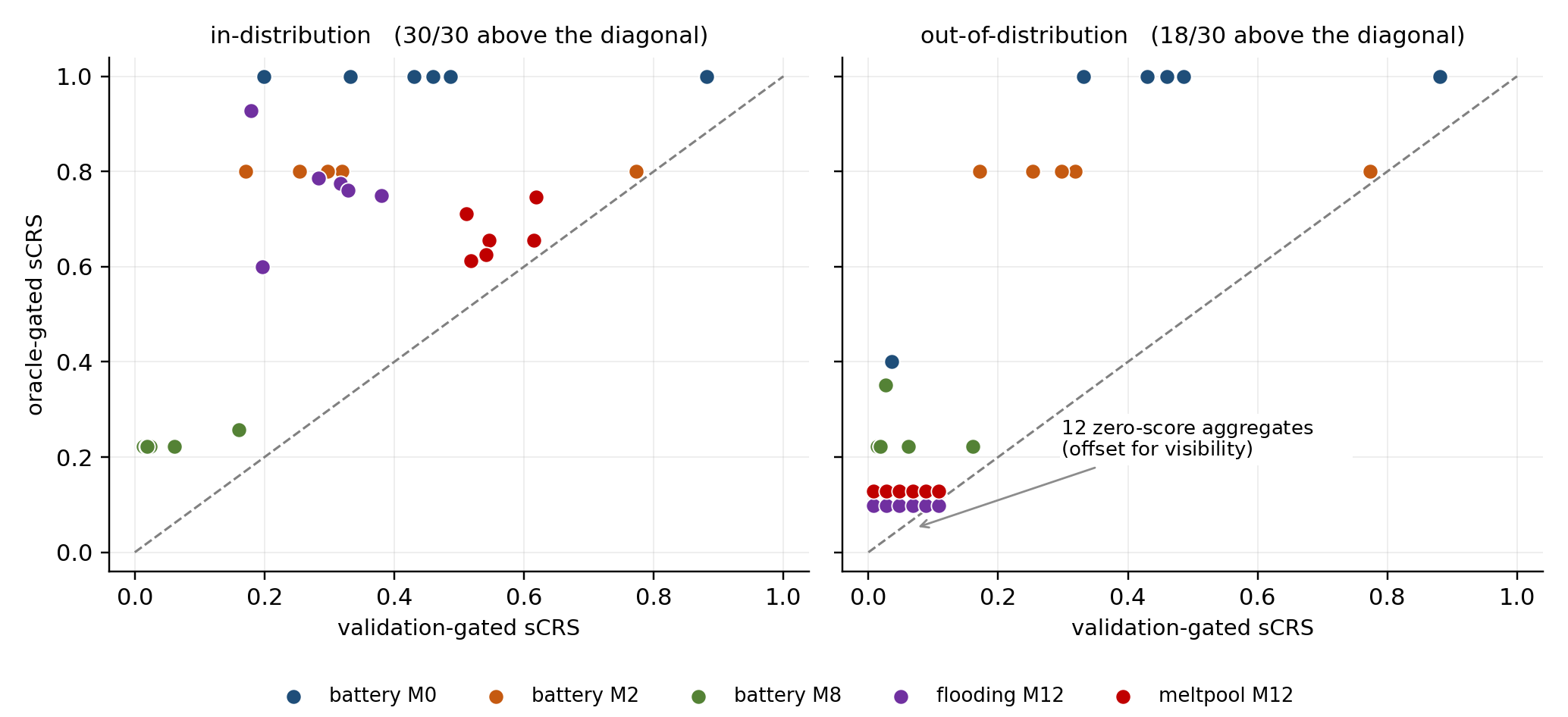}
\caption{Validation-error-gated versus Oracle-gated $\scrs$. Each point is the five-seed mean for one model--configuration combination; the dashed line indicates equal scores. The 12 zero-score points in the right panel are offset for visibility; their actual coordinates are all $(0,0)$.}
\label{fig:neural}
\end{figure}

Prediction averaging can reduce coefficient error without providing an effective abstention rule. Averaging five Poseidon-T models across the five configurations reduces macro-averaged OOD coefficient NRMSE from 0.164 to 0.139 (Appendix~\ref{app:poseidon}). This experiment evaluates prediction averaging, not direction selection based on ensemble spread. Its coefficient errors are also not directly comparable to the field errors in Table~\ref{tab:neural}.

\subsection{Agreement between Confidence Rankings and Observation Geometry}
We compare three confidence signals---calibration error, conformal interval width, and ensemble spread---prioritizing smaller values for each. No abstention, the geometric partition, and random ranking serve as references. Calibration errors and interval widths are obtained from calibration data and induce a fixed direction ranking for each configuration, without updating for individual test observations. This section evaluates interval widths as empirical ranking signals and does not test the statistical coverage of conformal intervals.

Table~\ref{tab:rules} summarizes 60 model--configuration--split combinations, where ``passing'' requires zero below-threshold claims and positive $\scrs$. The geometric partition passes in 47 combinations. Calibration-error selection passes in only 9 and makes below-threshold claims in 46, despite an average $\scrs$ of 0.375. A positive soft-penalty score therefore cannot replace direct inspection of directional claims. Ensemble spread and conformal interval width pass no combinations at the fixed operating point, but their ranking scores exceed the random references. Performance at a fixed operating point and ranking quality across budgets describe different properties.

\begin{table}[t]
\caption{Fixed-operating-point scores and ranking performance of direction-selection rules. Counts use a denominator of 60 (5 configurations $\times$ 6 model families $\times$ 2 splits); passing requires zero overclaiming and $\scrs>0$. Ranking columns report means and dispersion, not confidence intervals.}
\label{tab:rules}
\centering\small\setlength{\tabcolsep}{3pt}
\begin{tabular}{@{}lccccc@{}}
\toprule
Rule & $\scrs$ & \shortstack{Passing\\tasks} & \shortstack{Overclaiming\\tasks} & CRC-AUC & \shortstack{Geometric\\AUROC}\\
\midrule
No abstention & 0.000 & 0/60 & 60/60 & $0.000\pm0.000$ & $0.500\pm0.000$\\
Calibration error & 0.375 & 9/60 & 46/60 & $0.821\pm0.176$ & $0.649\pm0.139$\\
Ensemble spread & 0.000 & 0/60 & 0/60 & $0.758\pm0.191$ & $0.640\pm0.184$\\
Conformal width & 0.000 & 0/60 & 0/60 & $0.753\pm0.196$ & $0.742\pm0.161$\\
Geometric partition & 0.545 & 47/60 & 0/60 & $0.618\pm0.095$ & $1.000\pm0.000$\\
Random ranking & --- & --- & --- & 0.000 & 0.500\\
\bottomrule
\end{tabular}
\end{table}

Within the respective aggregation scopes of Table~\ref{tab:rules}, calibration error and conformal interval width yield CRC-AUC values of 0.821 and 0.753, and geometric AUROCs of 0.649 and 0.742. CRC-AUC treats successfully recovered directions as positives, whereas geometric AUROC treats retained directions as positives. They therefore evaluate rankings of successful recovery and geometric labels, respectively, and their values should not be compared directly.

Configuration-wise analysis shows that the agreement between conformal interval widths and geometric labels varies across configurations (Table~\ref{tab:geometric_auc}). Geometric AUROC is 0.909 for Battery M2, 0.868 for Melt-pool M12, 0.769 for Battery M8, and 0.589 for Flooding M12. The last value is only slightly above the random reference of 0.5, indicating a weaker correspondence between high-confidence and high-sensitivity directions.

\begin{table}[t]
\caption{Geometric AUROC for rankings by conformal interval width. Retained directions are positives; larger AUROC means they tend to receive higher confidence. The last column is the normalized excess area from Equation~\eqref{eq:excess}, which contains the same ranking information as AUROC.}
\label{tab:geometric_auc}
\centering\small
\begin{tabular}{@{}lcccc@{}}
\toprule
Configuration & $r_\tau$ & $b_\tau$ & $U_B$ & $2U_B-1$\\
\midrule
Battery M2 & 3 & 14 & 0.909 & 0.818\\
Battery M8 & 8 & 9 & 0.769 & 0.538\\
Melt-pool M12 & 9 & 8 & 0.868 & 0.736\\
Flooding M12 & 12 & 5 & 0.589 & 0.178\\
\bottomrule
\end{tabular}
\end{table}

\subsection{Recovery and Abstention across Claim Coverages}
\label{subsec:coverage}
An overall ranking score need not reflect behavior at a particular claim budget. We sweep claim coverage for all five configurations and examine recovery yield $Y$, the number of successfully recovered claimed directions divided by the total number of directions, and $Z$, the abstained fraction of below-threshold directions. Table~\ref{tab:areas} summarizes the excess abstention area $D$ over the random reference; Figures~\ref{fig:recovery} and~\ref{fig:abstention} show the two complete curves.

\begin{table}[t]
\caption{Excess area $D$ of below-threshold-abstention curves under ID conditions. Retained/below-threshold dimensions correspond to each learning configuration; hits count retained directions meeting the recovery tolerance. The last column is conformal-width $D$ divided by geometric-reference $D$.}
\label{tab:areas}
\centering\small\setlength{\tabcolsep}{3pt}
\begin{tabular}{@{}lccccccc@{}}
\toprule
Configuration & \shortstack{Retained/\\below} & Hits & \shortstack{Conformal\\$D$} & \shortstack{Ensemble\\$D$} & \shortstack{Calibration\\$D$} & \shortstack{Geometric\\$D$} & Ratio\\
\midrule
Battery M0 & 1/16 & 1 & 0.029 & 0.029 & 0.029 & 0.029 & 100\%\\
Battery M2 & 3/14 & 2 & 0.072 & 0.074 & 0.053 & 0.088 & 82\%\\
Battery M8 & 8/9 & 1 & 0.126 & 0.121 & 0.050 & 0.235 & 54\%\\
Melt-pool & 9/8 & 4--6 & 0.195 & 0.178 & 0.141 & 0.265 & 74\%\\
Flooding & 12/5 & 7--11 & 0.063 & 0.047 & 0.053 & 0.353 & 18\%\\
\bottomrule
\end{tabular}
\end{table}

Conformal interval width has positive excess area in all five configurations, but its gap from the geometric reference varies. Battery M0 has excess area 0.029, equal to the geometric reference; this remains a valid ranking comparison between one retained and 16 below-threshold directions. Flooding has excess area 0.063, only $18\%$ of the geometric reference of 0.353. The corresponding ratios for Battery M8 and Melt-pool are $54\%$ and $74\%$, showing that the ratio cannot be explained solely by the fraction of below-threshold directions.

\begin{figure}[t]
\centering
\includegraphics[width=\linewidth]{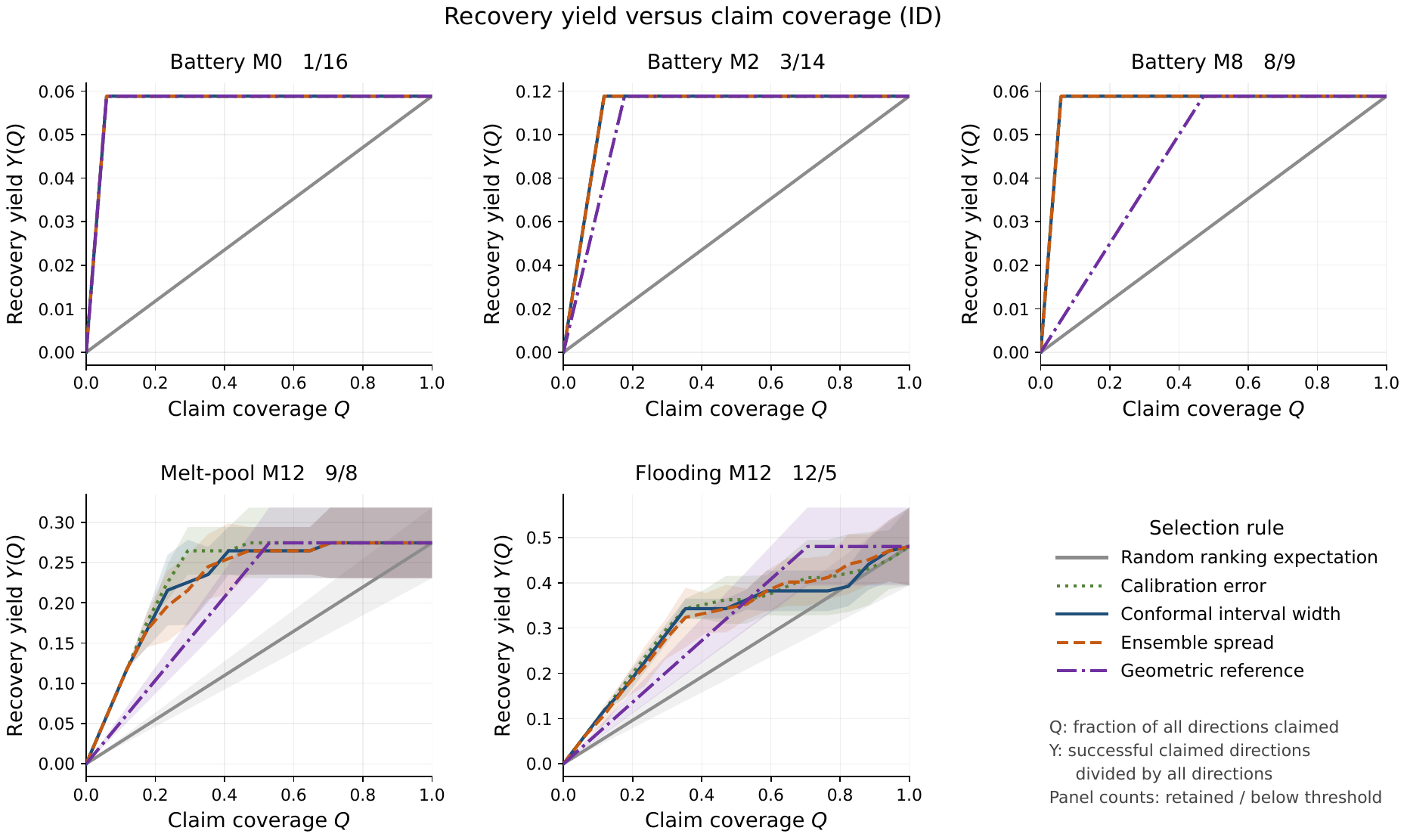}
\caption{Recovery yield $Y(Q)$ versus claim coverage $Q$ under ID conditions. The horizontal axis is the fraction of all directions claimed; the vertical axis is the number of successfully recovered claimed directions divided by all directions. Curves are means over six model families. Shading shows one standard deviation across model families, not a random-seed confidence interval. The gray straight line is the expected random-ranking curve under constant confidence, $Y(Q)=(k/m)Q$. Panel titles give retained/below-threshold dimensions; vertical scales differ across panels.}
\label{fig:recovery}
\end{figure}

\begin{figure}[t]
\centering
\includegraphics[width=\linewidth]{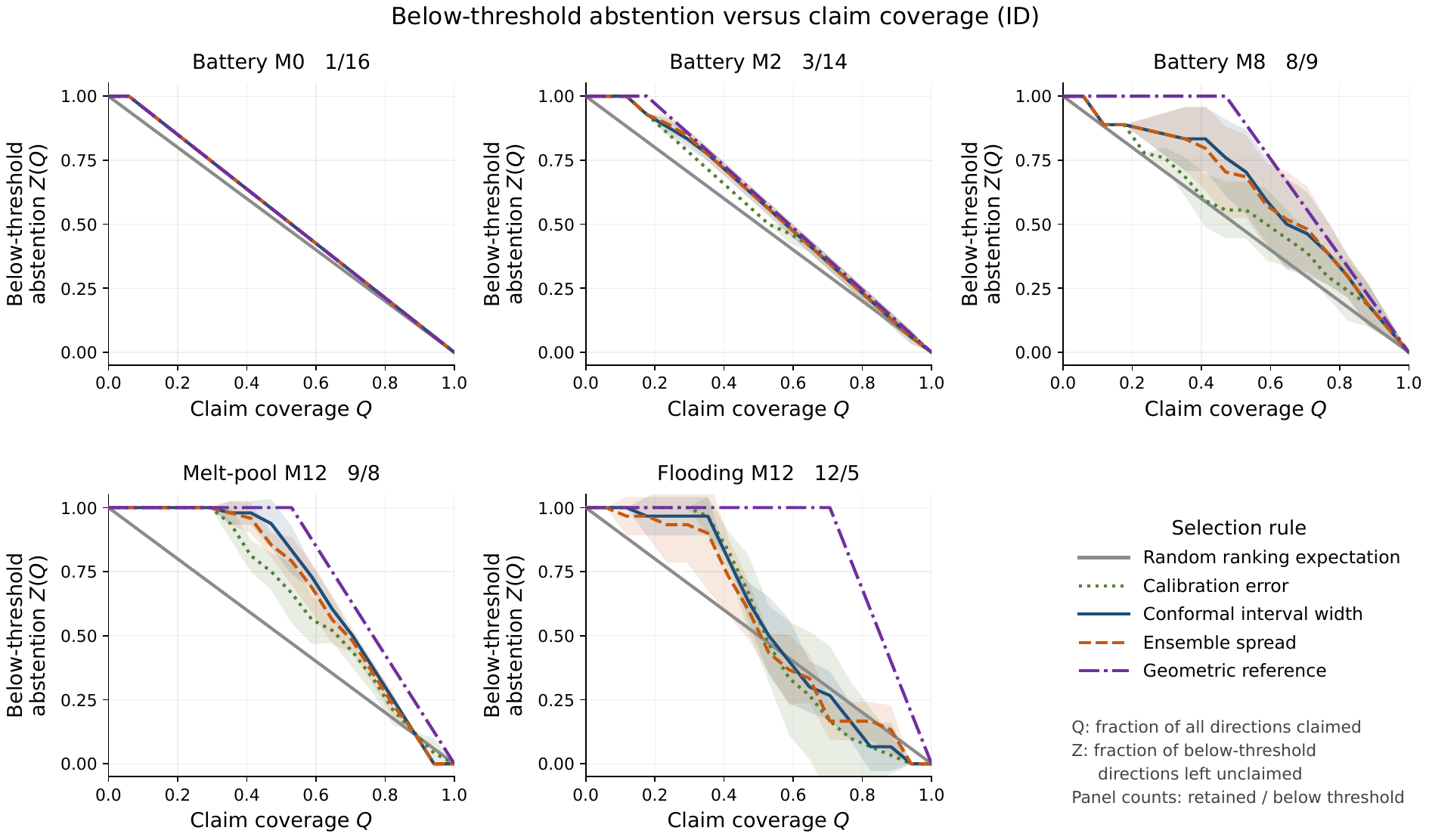}
\caption{Below-threshold abstention $Z(Q)$ versus claim coverage $Q$ under ID conditions. The horizontal axis is the fraction of all directions claimed; the vertical axis is the fraction of below-threshold directions left unclaimed. Curves and shading use the same aggregation as Figure~\ref{fig:recovery}. The gray line is the random reference $1-Q$. The geometric reference abstains from all below-threshold directions while the claim budget does not exceed the retained dimension. Panel titles give retained/below-threshold dimensions.}
\label{fig:abstention}
\end{figure}

Flooding M12 illustrates the distinction between overall ranking and behavior at particular budgets. Despite a geometric AUROC of 0.589, conformal-width ranking has mean below-threshold abstention below the random reference at all seven tested budgets from $Q=10/17$ to $16/17$ (Figure~\ref{fig:highcoverage}). At $Q=14/17\approx0.824$, this fraction is 0.067, compared with the random expectation $3/17\approx0.176$; at $Q=1$, both are zero. An overall AUROC above chance therefore does not guarantee higher below-threshold abstention at every coverage budget.

\begin{figure}[t]
\centering
\includegraphics[width=\linewidth]{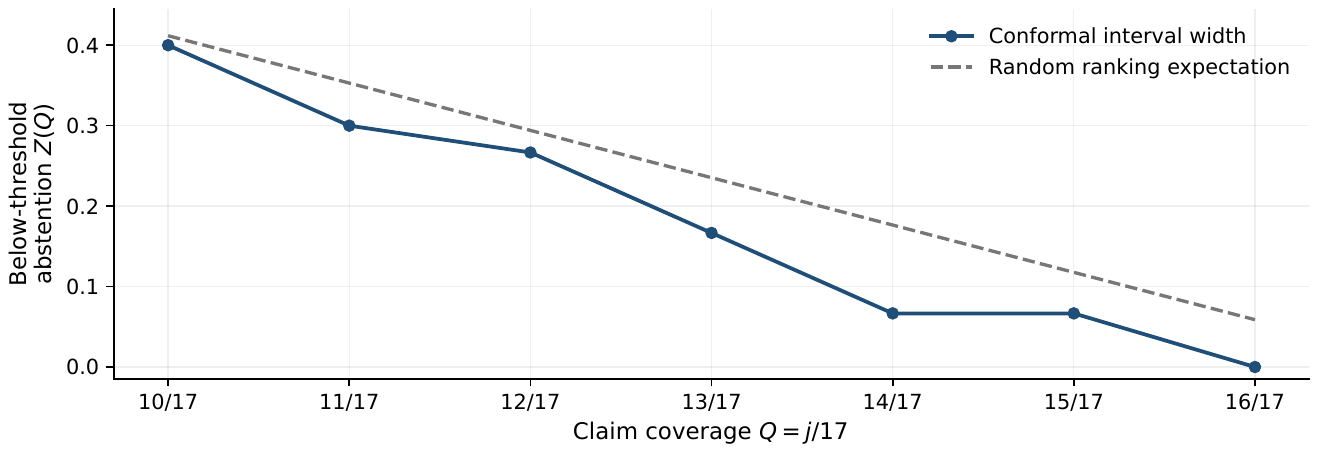}
\caption{Below-threshold abstention for Flooding M12 at seven high claim budgets. Points are means over six model families; the dashed line is the random expectation $1-Q$. Connecting lines guide the eye.}
\label{fig:highcoverage}
\end{figure}

One possible explanation is that prior correlations make some low-sensitivity directions highly predictable in distribution, leading to narrower intervals. The controlled inverse problems support the coexistence of prior-supported prediction and observation ambiguity, but do not establish whether this mechanism explains all ranking differences in the neural models. The current results support jointly reporting prediction recovery, geometric ranking, and budget-dependent abstention; they do not demonstrate failure of conformal coverage guarantees.

\FloatBarrier
\section{Conclusion}
\label{sec:conclusion}
We introduced UniPDE-Bench to evaluate local observation sensitivity, prediction recovery, and confidence rankings separately in a common directional basis, and established the correspondence between coverage-curve areas and AUROC. Results from three simulation operators and five sensing configurations show that low prediction error can coexist with a nontrivial below-threshold subspace, and that aggregate reconstruction accuracy cannot replace direction-selection assessment. The Flooding budget sweep further shows that an overall AUROC above chance can coexist with below-random below-threshold abstention at some claim budgets. Evaluation of partially observed PDE state estimation should therefore jointly report observation geometry, recovery yield, and budget-dependent abstention, rather than relying only on aggregate error or a single ranking score.
\label{maintext:end}

\section*{Reproducibility Statement}
Appendices~\ref{app:theory} and~\ref{app:auc} provide theoretical analysis and metric derivations; Appendices~\ref{app:numerics} and~\ref{app:additional} provide numerical diagnostics and supplementary experiments.

\section*{AI Use Statement}
Generative AI tools were used for manuscript review, Chinese-to-English translation, LaTeX preparation, and revisions to figure labels and captions.

\ifdefined\namedversion
\section*{Acknowledgements}
This work was supported by the National Natural Science Foundation of China (Grant Nos.\ 92267205, 92067205, 92367301, and 92267301), the Natural Science Foundation of Liaoning Province (Grant No.\ 2024-MSBA-83), the State Key Laboratory of Robotics of China (Grant No.\ 2023-Z15), and the National Program for Funded Postdoctoral Researchers (Grant No.\ GZB20230805).
\fi

\clearpage
\appendix
\section{Additional Theory for Local Observation Geometry}
\label{app:theory}
\subsection{Joint Ambiguity and Indistinguishability}
\paragraph{Proof of Proposition~\ref{prop:joint}.}
The two columns of $J=[1\;1]$ are nonzero, yet $(1,-1)^\top$ lies in its null space. More generally, the rank--nullity theorem gives $\dim\ker J=m-\rank(J)\geq m-p$. Nonzero coordinate-wise responses therefore do not guarantee that the joint linear map is injective.

\paragraph{Proof of Proposition~\ref{prop:indistinguishable}.}
When two states produce the same observation, any deterministic estimator returns the same estimate $\widehat c$. The triangle inequality gives
\[
\|c_+-c_-\|_G\leq\|c_+-\widehat c\|_G+\|\widehat c-c_-\|_G.
\]
At least one error is therefore no smaller than half the distance between the states. The result requires identical observations under the observation map itself, not merely a nontrivial Jacobian null space.

\begin{proposition}[Indistinguishability under bounded noise]
Suppose admissible noise satisfies $\|\Sigma_y^{-1/2}\eta\|_2\leq\beta$. If $c_+,c_-\in\mathcal C$ and
\[
\|\Sigma_y^{-1/2}(\mathcal H(c_+)-\mathcal H(c_-))\|_2\leq2\beta,
\]
then there are admissible noises associated with the two states that make their observations identical. The worst-case error of any deterministic estimator over these state--noise combinations is at least $\tfrac12\|c_+-c_-\|_G$.
\label{prop:bounded}
\end{proposition}
\begin{proof}
Choose the common observation $y_*=(\mathcal H(c_+)+\mathcal H(c_-))/2$. The whitened norm of the noise required for either state is at most $\beta$, so the triangle-inequality argument in Proposition~\ref{prop:indistinguishable} applies directly. A noise covariance does not itself imply bounded noise; resolution under unbounded noise, such as Gaussian noise, requires a probabilistic criterion.
\end{proof}

\subsection{Finite Perturbations and Local Linearization}
Suppose the Jacobian of the whitened observation map
\[
\widetilde{\mathcal H}(z)=\Sigma_y^{-1/2}\mathcal H(c_0+G^{-1/2}z)
\]
is $L$-Lipschitz on the relevant line segments. The linearization remainder then satisfies
\[
\|\widetilde{\mathcal H}(z)-\widetilde{\mathcal H}(0)-Az\|_2
\leq\frac L2\|z\|_2^2.
\]
Taking $z_\pm=\pm\delta v_i$ along a right singular direction gives
\[
\|\widetilde{\mathcal H}(z_+)-\widetilde{\mathcal H}(z_-)\|_2
\leq2\delta\sigma_i+L\delta^2.
\]
This bound requires admissible perturbed states and control of the Lipschitz constant and perturbation amplitude. We do not estimate these bounds here; the finite-perturbation results are used only as local sensitivity diagnostics.

\subsection{Observation Histories and Statistical Coverage}
\paragraph{Observation histories.}
With an unchanged state parameterization, if $O_{T+1}$ is obtained by appending observation rows to $O_T$, then $\ker O_{T+1}\subseteq\ker O_T$, so the exact rank cannot decrease. Relative effective rank, however, need not be monotone. At $\tau=0.05$, $A_0=\diag(1,0.1)$ has effective rank 2; appending the row $(100,0)$ changes its singular values to $\sqrt{10001}$ and $0.1$, reducing the effective rank to 1. This change results from normalization by the largest singular value, not from a loss of observation information.

\paragraph{Exchangeability.}
For independently and identically distributed state samples $(c_1,c_2)$, let $y=c_1$. The second coordinate is not identifiable from the observation over a two-dimensional admissible set, while the observation--target samples remain exchangeable. Marginal coverage guarantees therefore do not require the observation map to be injective, although distribution shift can invalidate these guarantees.

\section{Coverage-Curve Areas and AUROC}
\label{app:auc}
The derivations below use the classical equivalence between AUROC and the probability of correctly ordering a positive--negative pair, counting a tie as one half \citep{hanley}. They establish its relationship to the recovery-yield and below-threshold-abstention areas defined in this paper.

\subsection{Recovery-Coverage Area}
Let $t_1,\ldots,t_k$ be the ranks of the $k$ successfully recovered directions in descending confidence order, starting at 1. Each included hit increases recovery yield by $1/m$. Integrating the piecewise-linear curve in Equation~\eqref{eq:curves} gives
\[
\int_0^1Y(Q)\,dQ=\frac1{m^2}\sum_{a=1}^k\left(m-t_a+\frac12\right).
\]
The expected area under random ranking is $k/(2m)$. When all hits rank first, the area is $k/m-k^2/(2m^2)$. Let $L_h$ count direction pairs in which a non-hit precedes a hit. For $0<k<m$,
\[
\AUROC_h=1-\frac{L_h}{k(m-k)},\qquad
\int_0^1Y(Q)\,dQ=\frac km-\frac{k^2}{2m^2}-\frac{L_h}{m^2}.
\]
Substituting into Equation~\eqref{eq:crc} yields the linear relationship between normalized recovery area and the AUROC of hit labels. This score describes ranking quality rather than recovery quantity.

\subsection{Below-Threshold-Abstention Area}
\label{app:abstention_area}
Write $r=r_\tau$ and $b=b_\tau$, with both classes nonempty. A below-threshold direction at rank $t$ contributes $(t-\tfrac12)/(mb)$ to the area under the abstention curve. Summing its pairwise ordering relative to retained directions gives
\[
\int_0^1Z(Q)\,dQ=\frac b{2m}+\frac rmU_B.
\]
Subtracting the random-reference area $1/2$ yields $D=(r/m)(U_B-1/2)$. The geometric reference has $U_B=1$, so $D_{\mathrm{geom}}=r/(2m)$, proving Equation~\eqref{eq:excess}. These identities compare areas over the full curve; even when $U_B>1/2$, some budgets may satisfy $Z(Q)<1-Q$.

\subsection{Ties and Boundary Cases}
Taking the expectation over uniform random orderings within each tied group preserves the area identities above. The normalization denominator for recovery area is zero only when $k=0$ or $k=m$; one or two hits remain scorable within the nondegenerate range. Geometric AUROC requires both direction classes to be nonempty. Whenever a score is undefined, the curves, direction counts, and reason are still reported.

\section{Reference Spectra and Numerical Sensitivity}
\label{app:numerics}
\subsection{Singular-Value Spectra of Reference Configurations}
Table~\ref{tab:spectra} lists the normalized spectra underlying Figure~\ref{fig:spectra}, rounded to three decimal places. At $\tau=0.05$, the retained dimensions of Battery, Melt-pool, and Flooding are 12, 9, and 12. Determining exact rank requires the unrounded spectra and numerical rank tolerance.

\begin{table}[htbp]
\caption{Normalized singular-value spectra for the three reference configurations. Displayed zeros may result from padding or rounding and do not establish exact null-space dimensions.}
\label{tab:spectra}
\centering\small
\begin{tabular}{@{}cccc@{}}
\toprule
Index & Battery & Melt-pool & Flooding\\
\midrule
1 & 1.000 & 1.000 & 1.000\\
2 & 0.713 & 0.682 & 0.961\\
3 & 0.713 & 0.682 & 0.960\\
4 & 0.488 & 0.362 & 0.960\\
5 & 0.452 & 0.335 & 0.927\\
6 & 0.311 & 0.170 & 0.757\\
7 & 0.311 & 0.170 & 0.740\\
8 & 0.226 & 0.089 & 0.739\\
9 & 0.200 & 0.077 & 0.739\\
10 & 0.154 & 0.041 & 0.703\\
11 & 0.154 & 0.041 & 0.701\\
12 & 0.135 & 0.028 & 0.701\\
13 & 0.004 & 0.000 & 0.039\\
14--17 & 0.000 & 0.000 & 0.000\\
\bottomrule
\end{tabular}
\end{table}

\subsection{Grid Resolution and Finite-Difference Step Size}
Refining the grid from $48\times48$ to $64\times64$ leaves the effective ranks of Battery and Flooding at 12, while increasing Melt-pool from 7 to 9. This change indicates sensitivity of near-threshold singular directions to discretization; effective rank alone is insufficient to establish subspace convergence.

As the finite-difference step varies from $10^{-2}$ to $3\times10^{-4}$, the largest change in the Battery Jacobian norm between adjacent step sizes is $5.0\times10^{-6}$. Along the least sensitive direction, the ratios of observation change to shape change are $2.0\times10^{-7}$ for Battery, $9.2\times10^{-14}$ for Melt-pool, and $4.1\times10^{-13}$ for Flooding.

These norms and ratios support only local diagnostics. Establishing subspace convergence additionally requires comparing $\|J_h-J_{h/2}\|$, near-threshold singular values, and projection-matrix distances, and reporting perturbation amplitudes and absolute whitened observation differences.

\section{Supplementary Experimental Results}
\label{app:additional}
\subsection{High Claim Budgets in Flooding}
Table~\ref{tab:highcoverage} supplements the Flooding M12 ID results in Section~\ref{subsec:coverage}, covering claim budgets $j=10,\ldots,16$. The random reference uses the exact coverage $Q=j/17$. Abstention is below the random reference at all seven budgets, consistent with an overall geometric AUROC above chance (Appendix~\ref{app:abstention_area}).

\begin{table}[htbp]
\caption{High-claim-budget results for Flooding M12, averaged over six model families.}
\label{tab:highcoverage}
\centering\small
\begin{tabular}{@{}ccccc@{}}
\toprule
$j$ & $Q=j/17$ & Conformal $Z$ & Random $1-Q$ & Difference\\
\midrule
10 & 0.588 & 0.400 & 0.412 & $-0.012$\\
11 & 0.647 & 0.300 & 0.353 & $-0.053$\\
12 & 0.706 & 0.267 & 0.294 & $-0.027$\\
13 & 0.765 & 0.167 & 0.235 & $-0.069$\\
14 & 0.824 & 0.067 & 0.176 & $-0.110$\\
15 & 0.882 & 0.067 & 0.118 & $-0.051$\\
16 & 0.941 & 0.000 & 0.059 & $-0.059$\\
\bottomrule
\end{tabular}
\end{table}

\subsection{Poseidon-T Adaptation and Prediction Averaging}
\label{app:poseidon}
\paragraph{Experimental setup.}
The adapted Poseidon-T has 20.77M parameters, takes rasterized sensor histories as input, and predicts 17 interface coefficients. The dataset contains 576 trajectories. For each configuration, five random seeds are trained for a budget of 200 epochs with validation-based early stopping. We compare OOD coefficient NRMSE for individual models and the average prediction of five models. This experiment assesses point estimation, not backbone rankings under equal computational budgets.

\begin{table}[htbp]
\caption{Individual-model predictions and five-model prediction averaging for Poseidon-T. Single-model OOD coefficient NRMSE is the mean $\pm$ standard deviation over five seeds.}
\label{tab:poseidon}
\centering\small\setlength{\tabcolsep}{3pt}
\begin{tabular}{@{}lccccc@{}}
\toprule
Configuration & \shortstack{Single-model\\OOD NRMSE} & \shortstack{Averaged\\OOD NRMSE} & \shortstack{Relative\\change} & \shortstack{No abstention\\$\scrs$} & \shortstack{Oracle\\$\scrs$}\\
\midrule
Battery M0 & $0.149\pm0.007$ & 0.148 & $-0.7\%$ & 0.000 & 1.000\\
Battery M2 & $0.159\pm0.026$ & 0.156 & $-1.8\%$ & 0.000 & 0.840\\
Battery M8 & $0.148\pm0.026$ & 0.144 & $-2.5\%$ & 0.000 & 0.857\\
Melt-pool M12 & $0.193\pm0.040$ & 0.121 & $-37.1\%$ & 0.000 & 0.562\\
Flooding M12 & $0.171\pm0.056$ & 0.126 & $-26.4\%$ & 0.000 & 0.836\\
Macro-average & 0.164 & 0.139 & --- & --- & ---\\
\bottomrule
\end{tabular}
\end{table}

\paragraph{Results.}
Prediction averaging reduces OOD coefficient error in all five configurations, lowering macro-averaged NRMSE from 0.164 to 0.139. The reductions are larger for Melt-pool M12, from 0.193 to 0.121, and Flooding M12, from 0.171 to 0.126 (Figure~\ref{fig:poseidon}).

\begin{figure}[htbp]
\centering
\includegraphics[width=\linewidth]{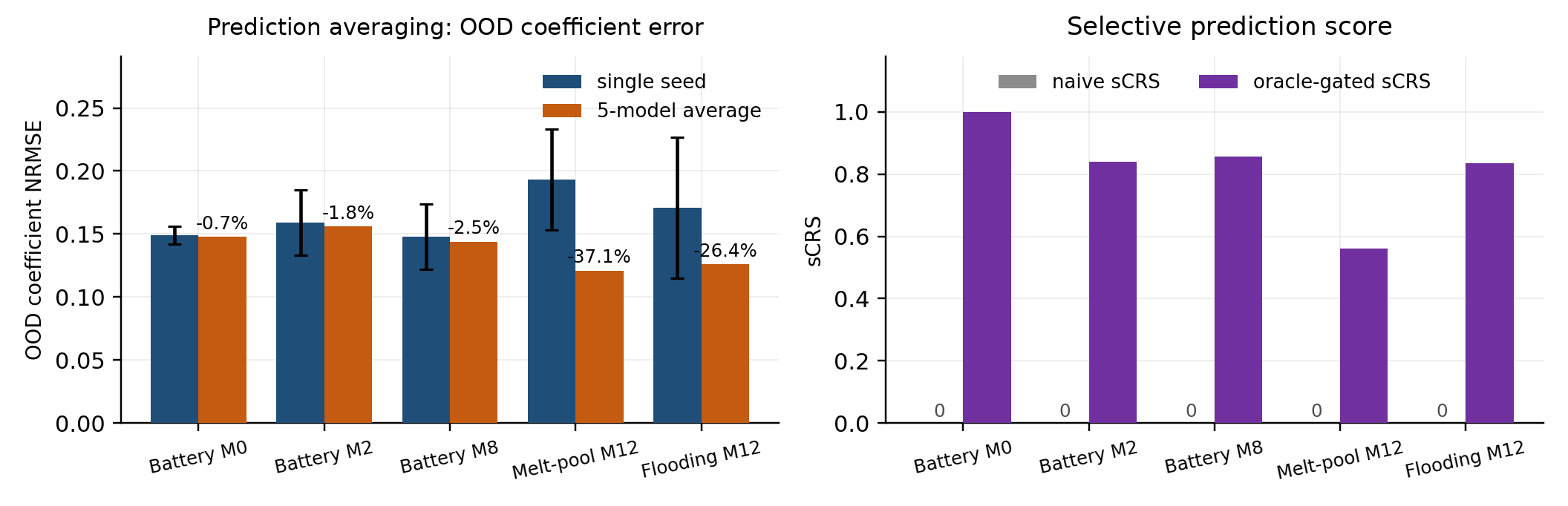}
\caption{Prediction averaging and selective scores for Poseidon-T. Left: single-model and five-model-averaged OOD coefficient NRMSE; error bars show the five-seed standard deviations of individual models. Right: $\scrs$ for no abstention and Oracle gating.}
\label{fig:poseidon}
\end{figure}

All 25 point-estimation models have zero no-abstention $\scrs$, reflecting their no-abstention behavior in configurations with below-threshold directions; this alone does not determine their recovery capability. Prediction averaging is also not equivalent to effective confidence-based selection. The latter requires a separate evaluation of ensemble spread using observation geometry that includes the same history inputs.

\end{document}